\documentclass[10pt]{article}
\usepackage{dcolumn}
\usepackage{bm}
\usepackage{verbatim}       

\usepackage[margin=4.3cm]{geometry}

\usepackage[dvips]{graphicx}

\usepackage{amsthm}

\usepackage{bm}
\usepackage{amstext}
\usepackage{psfrag}
\usepackage{amsmath}
\usepackage{mathtools}
\usepackage{amsfonts}
\usepackage{mathrsfs}
\newcommand{\p}{\partial}

\newcommand{\dd}{{\rm d}}

\renewcommand{\vec}[1]{\bm{#1}}

\newcommand{\bd}{\begin{definition}}                
\newcommand{\ed}{\end{definition}}                  
\newcommand{\bc}{\begin{corollary}}                 
\newcommand{\ec}{\end{corollary}}                   
\newcommand{\bl}{\begin{lemma}}                     
\newcommand{\el}{\end{lemma}}                       
\newcommand{\bp}{\begin{proposition}}            
\newcommand{\ep}{\end{proposition}}                
\newcommand{\bere}{\begin{remark}}                  
\newcommand{\ere}{\end{remark}}                     

\newcommand{\bt}{\begin{theorem}}
\newcommand{\et}{\end{theorem}}

\newcommand{\be}{\begin{equation}}
\newcommand{\ee}{\end{equation}}

\newcommand{\bit}{\begin{itemize}}
\newcommand{\eit}{\end{itemize}}
\newtheorem{theorem}{Theorem}[section]
\newtheorem{corollary}[theorem]{Corollary}
\newtheorem{lemma}[theorem]{Lemma}
\newtheorem{proposition}[theorem]{Proposition}
\theoremstyle{definition}
\newtheorem{definition}[theorem]{Definition}
\theoremstyle{remark}
\newtheorem{remark}[theorem]{Remark}

\begin{document}


\title{A geometrical introduction to screw theory}

\author{E. Minguzzi\thanks{
Dipartimento di Matematica Applicata ``G. Sansone'', Universit\`a
degli Studi di Firenze, Via S. Marta 3,  I-50139 Firenze, Italy.
E-mail: ettore.minguzzi@unifi.it} }


\date{}

\maketitle

\begin{abstract}
\noindent

This work introduces screw theory, a venerable but yet little known
theory aimed at describing rigid body dynamics. This formulation of
mechanics unifies in the concept of screw the translational and
rotational degrees of freedom of the body. It captures a remarkable
mathematical analogy between mechanical momenta and linear
velocities, and between forces and angular velocities. For instance,
it clarifies that angular velocities should be treated as applied
vectors and that, under the composition of motions, they sum with
the same rules of applied forces. This work provides a short and
rigorous introduction to screw theory intended to an undergraduate
and general readership.



\end{abstract}


\begin{flushleft}
\end{flushleft}

\tableofcontents

\section{Introduction}

The second law of Newtonian mechanics states that if $\vec{F}$ is
the force acting on a point particle of mass $m$ and $\vec{a}$ is
its acceleration, then $m \vec{a}=\vec{F}$. In a sense, the physical
meaning of this expression lies in its tacit assumptions, namely
that forces are vectors, that is, elements of a vector space, and as
such they sum. This experimental fact embodied in the second law is
what prevent us from considering the previous identity as a mere
definition of {\em force}.


Coming to the study of the rigid body, one can deduce the first
cardinal equation of mechanics $M \ddot{C}=\vec{F}$, where $C$ is
the affine point of the center of mass, $M$ is the total mass and
$\vec{F}^{ext}=\sum_i \vec{F}_i^{ext}$ is the resultant of the
external applied forces. This equation does not fix the dynamical
evolution of the body, indeed  one need to add the second cardinal
equation of mechanics $\dot{\vec{L}}(O)=\vec{M}(O)$, where
$\vec{L}(O)$ and $\vec{M}(O)$, are respectively, the total angular
momentum and the total mechanical momentum with respect to an
arbitrary fixed point $O$. Naively adding the applied forces  might
result in an incorrect calculation of $\vec{M}(O)$. As it is well
known, one must take into account the line of action of each force
$\vec{F}_i^{ext}$ in order to determine the {\em central axis},
namely the locus of allowed application points of  the resultant.

These considerations show that applied forces do not really form a
vector space. This unfortunate circumstance can be amended
considering, in place of the force, the field of mechanical momenta
that it determines (the so called {\em dynamical screw}). These type
of fields are constrained by the law which establishes the change of
the mechanical momenta under change of pole
\[
\vec{M}(P)-\vec{M}(Q)=\vec{F}\times (P-Q).
\]
An analogy between momenta and velocities, and between force
resultant and angular velocity is apparent considering the so called
{\em fundamental formula of the rigid body}, namely a constraint
which characterizes the velocity vector field of the rigid body
\[
\vec{v}(P)-\vec{v}(Q)=\vec{\omega} \times (P-Q).
\]
The correspondence can be pushed forward for instance by noting that
the concept of {\em instantaneous axis of rotation} is analogous to
that of {\em central axis}. Screw theory explores these analogies in
a systematic way and relates them to the Lie group of rigid motions
on the Euclidean space.

Perhaps, one of the most interesting consequences of screw theory is
that it allows us to fully understand that angular velocities should
be treated as  vectors applied to the instantaneous axis of
rotation, rather than as free vectors. This fact is not at al
obvious. Let us recall that the angular velocity is defined through
Poisson theorem, which states that, given a frame $K'$ moving with
respect to an absolute frame $K$, any normalized vector $\vec{e}'$
which is fixed with respect to $K'$ satisfies
\[
\frac{\dd \vec{e}'}{\dd t}=\vec{\omega} \times \vec{e}',
\]
in the original frame $K$, where $\vec{\omega}$ is unique. The
uniqueness allows us to  unambiguously define $\vec\omega$  as the
angular velocity of $K'$ with respect to $K$. As the vectors
$\vec{e}'$ are free, their application point is not fixed and so,
according to this traditional definition, $\vec{\omega}$ is not
given an application point.

This fact seems close to intuition. Indeed, let us consider
Foucault's 1851 famous experiment performed at the Paris
Observatory. By using a pendulum he was able to prove that the earth
 rotates with an angular velocity which coincides with that inferred from the observation of distant stars.
Of course, the choice of Paris was not essential, and the
measurement would have returned the same value for the angular
velocity were it performed in any other place on earth. In fact, the
reason for assigning to the angular velocity an application point in
the instantaneous axis of rotation becomes clear only in very
special applications, and in particular when the composition of
rigid motions is considered. This fact will be fully justified in
section \ref{njc}. Here we just wish to illustrate how, using the
analogy between forces and angular velocities, it is possible to
solve non-trivial problems on the composition of motions.

Consider, for instance, four frames $K_i$, $i=0,1,2,3$, where $K_0$
is the absolute frame and $K_{i+1}$, $i=0,1,2$,  moves with respect
to $K_i$ with an angular velocity $\vec{\omega}_{i,\, i+1}$. Let us
suppose that at the given instant of interest, and for  $i=0,1,2$,
the instantaneous axes of rotation of $K_{i+1}$ as seen from $K_i$,
lie all in the same plane as shown in figure \ref{fun}. We can apply
the well known rules of statics, for instance using the funicular
polygon \cite{timoshenko65,fasano78}, to obtain the angular velocity
$\vec{\omega}_{0,\, 3}$ and the instantaneous axis of rotation of
$K_3$ with respect to $K_0$.

%
%
%
%

\begin{figure}[ht]
\begin{center}
\includegraphics[width=9cm]{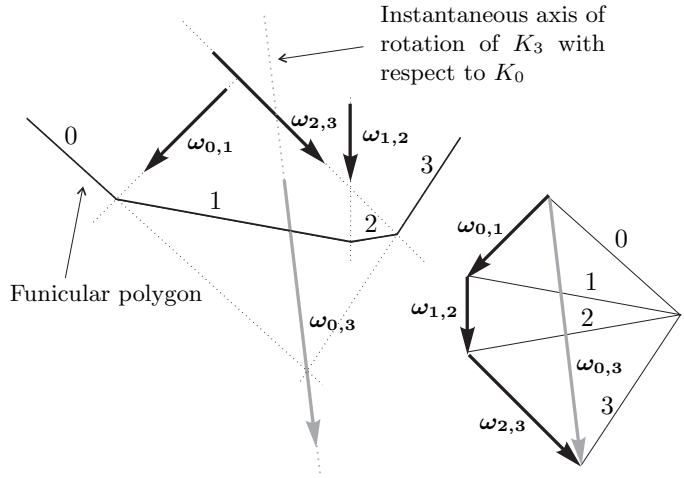}
\end{center}
\caption{Graphical determination of $K_3$ motion with respect to
$K_0$ by using the funicular polygon method. This method was
originally developed for finding the central axis in problems of
statics.} \label{fun}
\end{figure}

It is also interesting to observe if a frame $K_2$ rotates with
angular velocity $\vec{\omega}$ with respect to $K_1$, and $K_1$
rotates  with angular velocity $-\vec{\omega}$ with respect to
$K_0$, and if the two instantaneous axes of rotation are parallel
and separated by an arm of length $d$, then, at the given instant,
$K_2$ translates with velocity $\omega d$ in a direction
perpendicular to the plane determined by the two axes. As a
consequence, any act of translation can be reduced to a composition
of acts of pure rotation.

This result is analogous to the usual observation that two opposite
forces $\vec{F}$ and $-\vec{F}$ with arm $d$ generate a constant
mechanical momenta  of magnitude $dF$ and direction perpendicular to
the plane determined by the two forces. As a consequence, any
applied momenta can be seen as the effect of a couple of forces.

Of course, screw theory has other interesting consequences and
advantages. We invite the reader to discover and explore them in the
following sections.

The key ideas leading to screw theory included in this article have
been taught at a second year undergraduate course of ``Rational
Mechanics'' at the Faculty of Engineer of Florence University (saved
for the last technical section). We shaped this text so as to be
used  by our students for self study and by any other scholars who
might want to introduce screw theory in an undergraduate course.
Indeed, we believe that it is time to introduce this beautiful
approach to mechanics already at the level of undergraduate
University programs.

\subsection{Comments on previous treatments}

Screw theory is venerable (for an account of the early history see
\cite{dimentberg68}). It originated from the works of Euler, Mozzi
and Chasles, who discovered that any rigid motion can be obtained as
a rotation followed by a translation in the direction of the
rotation's axis (this is the celebrated Chasles's theorem which was
actually first obtained by Giulio Mozzi \cite{ceccarelli00}), and by
Poinsot, Chasles and M\"obius, who noted the analogy between forces
and angular velocities for the first time \cite{dimentberg68}.

It was developed and  reviewed by Sir R. Ball in his 1870 treatise
\cite{ball76}, and further developed, especially in connection with
its algebraic formulation, by Clifford, Kotel'nikov, Zeylinger,
Study and others. Unfortunately, by the end of the nineteenth
century it was essentially forgotten to be then fully rediscovered
only in the second half of the twentieth century. It remains largely
unknown and keeps being rediscovered by various authors interested
in rigid body mechanics (including this author).



Unfortunately, screw theory is usually explained  following
descriptive definitions rather than short axiomatic lines of
reasoning.
As a result, the available introductions are somewhat unsatisfactory
to the modern mathematical and physical minded reader. Perhaps for
this reason, some authors among the few that are aware of the
existence of screw theory claim that it is too complicated to
deserve to be taught. For instance, the last edition of Goldstein's
textbook \cite{goldstein01} includes a footnote which, after
introducing the full version of Chasles' theorem (Sect. \ref{mvj}),
which might be regarded as the starting point of screw theory,
comments

\begin{quote}
[$\ldots$] there seems to be little present use for this version of
Chasles' theorem, nor for the elaborate mathematics of screw motions
elaborated at the end of the nineteenth century.
\end{quote}

Were it written in the fifties of the last century this claim could
have been shared, but further on screw theory has become a main tool
for robotics \cite{jazar10} where it is ordinarily used.
Furthermore, while elaborate the mathematics of screw theory
simplifies the development of mechanics. Admittedly, however, some
people could be dissatisfied with available treatments and so its
main advantages can be underestimated. We offer here a shorter
introduction which, hopefully, could convince these readers of
taking a route into screw theory.


Let us comment on some definitions of screw that can be found in the
literature, so as to justify our choices.

A first approach, that this author does not find  appealing,
introduces the screw by means of the concept of {\em motor}. This
formalism depends on the point of reduction, and one finds the added
difficulty of proving the independence of the various deductions
from the chosen reduction point. It hides the geometrical content of
the screw and makes proofs lengthier. Nevertheless, it must be said
that the motor approach could be convenient for reducing screw
calculations to a matter of algebra (the so called screw calculus
\cite{dimentberg68}).

In a similar vein,  some references, including  Selig's
\cite{selig05}, introduce the screw from a matrix formulation that
tacitly assumes that a choice of reference frame has been made (thus
losing the invariance at sight of the definition).

Still concerning the screw definition,  some literature  follows the
practical and traditional approach which introduces the screw from
its properties (screw axis, pitch, etc.)
\cite{ball76,hunt90,featherstone08}, like in old fashioned linear
algebra where one would have defined a vector from its direction,
verse and module, instead of defining it as an element of a vector
space (to complicate matters, some authors define the screw up to a
positive constant, in other words they work with a projective space
rather that a vector space).  This approach could be more intuitive
but might also give a false confidence of understanding, and it is
less suited for a formal development of the theory. It is clear that
the vector space approach in linear algebra, while less intuitive at
first, proves to be much more powerful than any descriptive
approach. Of course, one has to complement it with the descriptive
point of view in order to help the intuition. In my opinion the same
type of strategy should be followed in screw theory, with a maybe
more formal introduction, giving a solid basis, aided by examples to
help the intuition. Since descriptive intuitive approaches are not
lacking in the literature, this work aims at giving a short
introduction of more abstract and geometrical type.

It should be said that at places there is an excess of formality in
the available presentations of screw theory. I refer to the tendency
of giving separated definitions of screws, one for the kinematical
{\em twist} describing the velocity field of the rigid body, and the
other for the {\em wrench} describing the forces acting on the body.
This type of approach, requiring definitions for screws and their
dual elements  (sometimes called co-screws),  lengthens the
presentation and forces the introduction and use of the  dual space
of a vector space, a choice which is not so popular especially for
undergraduate teaching.

Who adopts this point of view argues that  it should also be adopted
for forces in mechanics, which should be treated as 1-forms instead
as vectors. This suggestion, inspired by the concept of conjugate
momenta of Lagrangian and Hamiltonian theory, sounds more modern,
but would be geometrically well founded only if one could develop
mechanics without any mention to the scalar product.
 The scalar product allows us to identify a vector space with its dual
and hence to work only with the former. If what really matters is
the pairing between a vector space and its dual then, as this makes
sense even without scalar product, we could dispense of it. It is
easy to realize that in order to develop mechanics we need a vector
space (and/or its dual) as well as  a scalar product and an
orientation (although most physical combinations of interest might
be rewritten so as a to get rid of it, e.g.\ the kinetic energy is
$T=\frac{1}{2} \bf{p}[\bf{v}]$).

Analogously, in screw theory,  it could seem more appealing to look
at kinematical twists as screws and to dynamical wrenches as
co-screws, but geometrically this choice does not seem compelling,
and in fact it is questionable, given the price to be paid in terms
of length and loss of unity of the presentation. Therefore, we are
going to use just one mathematical entity - the screw - emphasizing
the role of the screw scalar product in identifying screws and dual
elements.

%



In this work I took care at introducing screw theory in a way as far
as coordinate independent as possible, but avoiding the traditional
descriptive route. In this approach the relation with the Lie
algebra of rigid maps becomes particularly transparent. Finally,
most approaches postpone the definition of screw after the examples
of systems of applied forces from which the idea of screw can been
derived. I think that it is better to introduce the screw first and
then to look at the applications.

In this way, through some key choices, I have obtained a hopefully
clear and straightforward introduction to screw theory, which is at
the same time mathematically rigorous. My hope is that after reading
these notes, the reader will share the author's opinion that screw
theory is indeed ``the right'' way of teaching rigid body mechanics
as the tight relation with the Lie group of rigid maps suggests.


\section{Abstract screw theory}

In this section we define the {\em screw} and prove some fundamental
properties. Specific applications will appear in the next sections.

Let us denote with $E$ the affine Euclidean space modeled on the
three dimensional vector space $V$. The space $V$ is endowed with a
positive definite scalar product $\cdot: V\times V\to \mathbb{R}$,
and is given an orientation. This structure is represented with a
triple $(V,\cdot,o)$ where $o$ denotes the orientation. Note that
thanks to this structure a vector product $\times: V\times V\to V$
can be naturally defined on $V$. Points of $E$ are denoted with
capital letters e.g.\ $P,Q,\ldots$ while points in $V$ are denoted
as $\vec{a}, \vec{b}, \ldots$ We shall repeatedly use the fact that
the mixed product $\vec{a} \cdot (\vec{b}\times \vec{c})$ changes
sign under odd permutations of its terms and remains the same under
even permutations. A vector field is a map $ f: E \to V$.

An {\em applied vector} is an element of $E\times V$, namely a pair
$(Q,\vec{v})$ where $Q$ is the {\em application point} of
$(Q,\vec{v})$. A {\em sliding vector} is an equivalence class of
applied vectors, where two applied vectors $(Q,\vec{v})$  and
$(Q',\vec{v}')$  are equivalent if $\vec{v}=\vec{v'}$ and for some
$\lambda \in \mathbb{R}$, $Q'-Q=\lambda \vec{v}$, namely they have
the same line of action. We shall preferably use the concept of
applied vector even in those cases in which it could be equivalently
replaced by that of sliding vector. The reason is that the concept
of sliding vector is superfluous because it is more convenient to
regard applied and sliding vectors as special types of screws.


Occasionally, we shall use the concept of {\em reference frame}
which is defined by a  choice of {\em origin} $O\in E$, and of
positive oriented orthonormal base
$\{\vec{e}_1,\vec{e}_2,\vec{e}_3\}$  for $(V,\cdot,o)$. Once a
reference frame has been fixed, any point $P\in E$ is univocally
determined by its coordinates $x^i, i=1,2,3$, defined through the
equation $P=O+\sum_i x^i \vec{e}_i$.

\begin{remark}
In order to lighten the formalism we shall consider different
physical vector quantities, such as position, velocity, linear
momenta, force, mechanical momenta, as elements of the same vector
space $V$. A more rigorous treatment would introduce a different
vector space for each one of these concepts. The reader might
imagine to have fixed the dimension units. It is understood that,
say, a linear momenta cannot be summed to a force even though in our
treatment they appear to belong to the same vector space.
\end{remark}

\begin{definition}
A {\em screw} is a vector field $ s: E \to V$ which admits some $
\vec{s} \in V$ in such a way that for any two points $P,Q\in E$
\begin{equation} \label{nuv}
s(P)-s(Q)=\vec{s} \times (P-Q).
\end{equation}
\end{definition}
For any screw $s$ the vector $\vec{s}$ is unique, indeed if
$\vec{s}$ and  $\vec{s}\,'$ satisfy the above equation, then
subtracting the  corresponding equations
$(\vec{s}\,'-\vec{s})\times(P-Q)=0$ and from the arbitrariness of
$P$, $\vec{s}\,'=\vec{s}$. The vector $\vec{s}$ is called the {\em
resultant of the screw}. If the resultant of the screw vanishes then
$s(P)$ does not depend on $P$ and  the screw is said to be {\em
constant}. Equation (\ref{nuv}) is the {\em constitutive equation of
the screw}.

\begin{definition}
If $s$ is a screw the quantity $s(P)\cdot \vec{s}$ does not depend
on the point and in called the {\em  scalar invariant of the screw}.
The {\em vector invariant  of the screw} is the quantity
(independent of $P$) and defined by
\begin{align*}
\bm{\mathcal{s}}&=s(P), \qquad &if \ \vec{s}=\vec{0}, \\
\bm{\mathcal{s}}&=\frac{s(P)\cdot \vec{s}}{{\vec{s}\cdot\vec{s}}}
\,{\vec{s}},  \qquad &if  \ \vec{s}\ne \vec{0}.
\end{align*}
%
%
%
%
\end{definition}
 Thus  if $\vec{s}\ne \vec{0}$ the vector
invariant of the screw is the projection of $s(P)$ on the direction
given by the resultant, and it is actually independent of $P$.

\begin{proposition}
The screws form a vector space $S$ and the map which sends $s$ to
$\vec{s}$ is linear.
\end{proposition}

\begin{proof}
If $s_1$ and $s_2$ are two screws
\begin{align*}
s_1(P)-s_1(Q)&=\vec{s}_1 \times (P-Q), \\
s_2(P)-s_2(Q)&=\vec{s}_2 \times (P-Q).
\end{align*}
Multiplying by $\alpha$ the first equation and adding the latter
multiplied by $\beta$ we get
\begin{equation} \label{hwx}
(\alpha s_1+\beta s_2)(P)-(\alpha s_1+\beta s_2)(Q)=(\alpha
\vec{s}_1+\beta \vec{s}_2) \times (P-Q),
\end{equation}
which implies that the screws form a vector space and that the
resultant of the screw $\alpha s_1+\beta s_2$ is $\alpha
\vec{s}_1+\beta \vec{s}_2$, that is, the map $s \to \vec{s}$ is
linear.
\end{proof}

Given two screws $s_1$ and $s_2$ let us consider the quantity
 \[\langle s_1, s_2\rangle(P):=\vec{s}_1 \cdot s_2(P) +
\vec{s}_2 \cdot s_1(P).\]

\begin{proposition}
For any two points $P,Q\in E$, $\langle s_1, s_2\rangle(P)=\langle
s_1, s_2\rangle(Q)$.
\end{proposition}

\begin{proof}
By definition $s_1(P)-s_1(Q)=\vec{s}_1 \times (P-Q)$ and
$s_2(P)-s_2(Q)=\vec{s}_2 \times (P-Q)$, thus
\begin{align*}
\vec{s}_1 \cdot s_2(P) + &\vec{s}_2 \cdot s_1(P)=\vec{s}_1 \cdot
(s_2(Q)+\vec{s}_2 \times (P-Q))+ \vec{s}_2 \cdot (s_1(Q)+\vec{s}_1
\times (P-Q))\\
&=\vec{s}_1 \cdot s_2(Q)+ \vec{s}_2 \cdot s_1(Q)+\{\vec{s}_1
\cdot[\vec{s}_2 \times (P-Q)]+\vec{s}_2 \cdot[\vec{s}_1 \times
(P-Q)]\}\\
&=\vec{s}_1 \cdot s_2(Q)+ \vec{s}_2 \cdot s_1(Q).
\end{align*}
\end{proof}

According to the previous result we can simply write $\langle s_1,
s_2\rangle$ in place of $\langle s_1, s_2\rangle(P)$.

\begin{definition} The {\em screw scalar product} is
the symmetric bilinear map $\langle \cdot ,\cdot \rangle: S\times
S\to \mathbb{R}$ which sends $(s_1,s_2)$ to  $\langle s_1,
s_2\rangle$.
\end{definition}

Note that the  scalar invariant of a screw is one-half the screw
scalar product of the screw by itself. Since this scalar invariant
can be negative, the screw scalar product on $S$ is not positive
definite. Nevertheless, we shall see that it is non-degenerate
(Sect. \ref{gvo}).

The cartesian product $V \times V$ endowed with the usual sum and
product by scalar gives the direct sum $V\oplus V$. Typically, there
will be three ways to construct screws out of (applied) vectors. The
easy proofs to the next two propositions are left to the reader.

\begin{proposition}
The map $\alpha: V \to S$ given by $\vec{v}\to s(P):=\vec{v}$ sends
a (free) vector to a constant screw. The map $\beta: E\times V \to
S$ given by $(Q,\vec{w})\to s(P):=\vec{w}\times (P-Q)$ sends an
applied vector to a screw. The map $\gamma: E\times V\times V\to S$
given by  $((Q,\vec{w}),\vec{v}) \to s(P):=\vec{v}+\vec{w}\times
(P-Q)$ sends a pair given by an applied vector and a free vector to
a screw.
\end{proposition}

The screws in the image of $\alpha$ will be called constant or {\em
free screws}. The screws in the image of $\beta$ will be called {\em
applied screws}. Clearly, by the constitutive equation of the screw,
the map $\gamma$ is surjective. In particular, every screw is the
sum of a free screw and an applied screw.

\begin{proposition}
Let $\gamma_O=\gamma(O,\cdot,\cdot): V\oplus V \to S$, then this
linear map is bijective.
%
\end{proposition}

Its inverse $\gamma_O^{-1}:S\to V\oplus V$ is called {\em motor
reduction} at $O$.
%
Once we agree on the reduction point $O$, any pair $(\vec{s}, s(O)
)$ as in the previous proposition is  called a {\em motor} at $O$.
Sometimes we shall write $\vec{s}^{O}$ for $s(O)$, thus the motor at
$O$ reads $(\vec{s}, \vec{s}^{O})$. Often, for reasons that will be
soon clear, we will prefer to represent the ordered pair in a column
form of two elements of $V$.

We can write the found bijective correspondence between $S$ and
$V^2$ as follows
\[
s \in S \xleftrightarrow{origin \ O} \begin{pmatrix} \vec{s}\\
\vec{s}^{O}\end{pmatrix} \in V\oplus V.
\]
In this representation the screw scalar product is given by $\langle
s_1, s_2\rangle = \vec{s}_1 \cdot \vec{s}_2^{\,O} + \vec{s}_2 \cdot
s_1^{O}$, thus is is mediated by the matrix $\begin{pmatrix} 0 & I
\\ I & 0 \end{pmatrix}$ where $I: V\to V$ is the identity map.

Let us now recall that any point $O\in E$ can be used as {\em
origin}, namely it allows us to establish a bijective correspondence
between $E$ and $V$ given by $P\to P-O$. If we additionally
introduce a positive oriented orthonormal base then we further have
the linear isomorphism $V \xleftrightarrow{base} \mathbb{R}^3$,
thus, as a result, given a full reference frame the screw gets
represented by an element of $\mathbb{R}^6$ in which the first three
components are those of $\vec{s}$ while the last three components
are those of $\vec{s}^{O}$.


\begin{definition}
Given a screw $s\in S$, the {\em screw axis} of $s$ is the set of
points for which the screw field has minimum module.
\end{definition}

\begin{proposition}
The screw axis coincides with the set $E$ if $\vec{s}=\vec{0}$ and
with a line of direction $\vec{s}$ if $\vec{s}\ne \vec{0}$. In both
cases, if $Q$ is any point in the screw axis then
\begin{equation} \label{njf}
s(P)=\bm{\mathcal{s}}+\vec{s}\times (P-Q).
\end{equation}
As a consequence, the screw axis is the set of points for which the
screw field coincides with $\bm{\mathcal{s}}$. For any point $Q$ on
the axis the motor reduction at $Q$ is $\vec{s} \oplus
\bm{\mathcal{s}}$.
\end{proposition}

%
%

Let us observe that the former term in the right-hand  side of Eq.
(\ref{njf}) is proportional to $\vec{s}$ and independent of the
point, while the latter term is orthogonal to to $\vec{s}$ and
dependent on the point.

\begin{proof}
Let us suppose $\vec{s}\ne \vec{0}$, the other case being trivial.
Let $A$ be any point, then it is easy to check that the axis which
passes through $Q$ in direction $\vec{s}$ where
\begin{equation} \label{nxc}
Q=A+\frac{\vec{s}\times s(A)}{\vec{s}\cdot \vec{s}},
\end{equation}
is made of points $R$ for which $s(R)=\bm{\mathcal{s}}$. Using the
constitutive equation of screws we find that Eq. (\ref{njf}) holds.
If $P$ is another point for which $s(P)=\bm{\mathcal{s}}$ then that
same equation gives $\vec{s}\times (P-Q)=\vec{0}$, which implies
that $P$ stays in the axis. Thus the found axis is the locus of
points $P$ for which $s(P)=\bm{\mathcal{s}}$. Equation (\ref{njf})
and the fact that $\bm{\mathcal{s}}\propto \vec{s}$ imply that this
axis is made of points for which the screw field is minimal. The
other claims follow easily.
\end{proof}

\begin{remark}
Usually the vector invariant and the screw axis are defined only for
$\vec{s}\ne \vec{0}$. However, we observe that it is convenient to
extend the definition as done here in such a way that Eq.
(\ref{njf}) holds for any screw. The case $\vec{s}=\vec{0}$ is
admittedly special and can be called {\em degenerate}.
\end{remark}

\begin{remark} \label{nhs}
The composition of applied vectors is nothing but the addition of
the corresponding screws in the vector space $S$. The resultant
screw can then be represented with its motor in  the screw axis
which is given by the resultant $\vec{s}$ aligned with the axis and
the invariant vector $\bm{\mathcal{s}}$ having the same direction
(Fig. \ref{fug}). In this sense the composition of applied vectors
does not give an applied vector. The operation of composition is
closed only if the full space of screws is considered.
\end{remark}

\begin{figure}[ht]
\begin{center}
 \includegraphics[width=9cm]{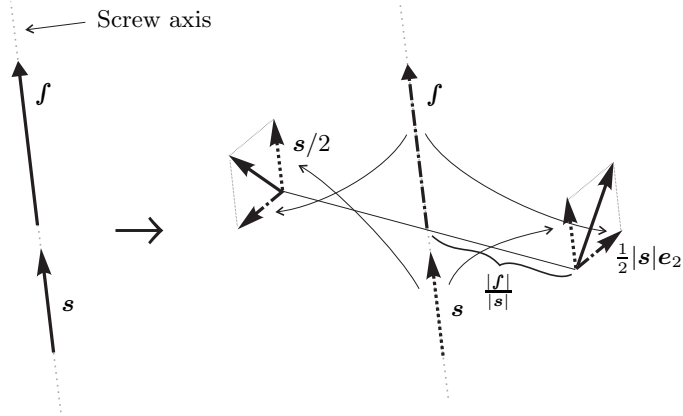}
\end{center}
\caption{Reduction of a screw to the simplest system of applied
screws (case $\vec{s}\ne\vec{0}$).} \label{fug}
\end{figure}

\begin{remark}
Two systems of applied vectors are said to be {\em equivalent} if
they determine the same screw. One often looks for the simplest way
of representing a screw by applied vectors. This is accomplished as
follows. The screw is the sum of the free screw given by
$\bm{\mathcal{s}}$ and the applied screw determined by the resultant
$\vec{s}$ applied on the screw axis. The free screw is generated by
two opposite applied vectors $\vec{s}_2$, $-\vec{s}_2$, placed in a
plane perpendicular to $\bm{\mathcal{s}}$ and such that their
magnitude times their arm gives ${\mathcal{s}}$. This reduces any
screw to the sum of at most three applied screw (two if
$\vec{s}=\vec{0}$). If $\vec{s}\ne \vec{0}$ the number can be
reduced to two regarding $\vec{s}$ as the sum
$\frac{1}{2}\vec{s}+\frac{1}{2}\vec{s}$, and absorbing one term of
type $\frac{1}{2}\vec{s}$ through a redefinition of  $\vec{s}_2$,
and analogously for the other  (see Fig.\ \ref{fug}). The arm can be
chosen in such a way that the resultants of the two applied screws
are perpendicular. In summary any screw is generated by two applied
screws whose resultants are either opposite with screw axes
belonging to the same plane (if $\vec{s}=\vec{0}$), or equal in
magnitude and perpendicular (if $\vec{s}\ne \vec{0}$).

\end{remark}
%
%

\begin{definition}
The {\em pitch} $p\in \mathbb{R}$ of a screw $s$, with $\vec{s}\ne
0$, is that constant such that $\bm{\mathcal{s}}=\frac{p}{2\pi}
\vec{s}$. If $\vec{s}=\vec{0}$ and $\bm{\mathcal{s}} \ne \vec{0}$,
we set by definition $p=+\infty$.
\end{definition}

Clearly, for a non-trivial screw, the pitch vanishes if and only if
the screw is an applied screw, and the pitch equals $+\infty$ if and
only if the screw is a free screw. The screws with a given pitch do
not form a vector subspace.

\begin{remark}
Using the pitch the screw can be rewritten
\[
s(P)=a [\frac{p}{2\pi}\, \vec{e}+\vec{e}\times (P-Q)],
\]
where $\vec{s}=a \vec{e}$, with $\vec{e}$ normalized vector and $a
\ge 0$. The quantity $a$ is called {\em amplitude of the screw}. It
must be said that for Sir R. S. Ball  \cite{ball76}  the screw is
$s/a$. However, it is not particularly convenient to regard $s/a$ as
a fundamental object since these type of normalized screws do not
form a vector space. Sir R. S. Ball would refer to our screws as
{\em screw motions}. We prefer to use our  shorter terminology
(shared by \cite{selig05}) because, for a dynamical screw $d$, which
we shall later introduce, no actual motion needs to take place. Note
also that the normalization of the pitch is chosen in such a way
that, integrating the screw vector field by a parameter $2\pi$,
i.e.\ by making a full rotation, one gets a diffeomorphism which is
a translation by $p$ along the screw axis. In other words, with the
chosen normalization, the pitch gives the translation of the screw
for any full rotation.
\end{remark}

\subsection{The commutator}

Every screw is a vector field, thus we can form the Lie bracket
$[s_1,s_2]$ of two screws \cite{kobayashi63}. In this section we
check that this commutator is itself a screw and calculate its
resultant.

\begin{proposition} \label{mjc}
The Lie bracket  $s=[s_1,s_2]$ is a screw with resultant $\vec{s}=-\vec{s}_1\times \vec{s_2}$ and satisfies
\begin{equation} \label{jfp}
s(P)=\vec{s}_2 \times s_1(P)-\vec{s}_1 \times s_2(P).
\end{equation}.
\end{proposition}

\begin{remark}
Some authors define the commutator of two screws as minus the Lie
bracket.
\end{remark}

\begin{proof}
Let $s_1$ and $s_2$ be two screws
\begin{align}
s_1(P)-s_1(Q)&=\vec{s}_1 \times (P-Q), \\
s_2(P)-s_2(Q)&=\vec{s}_2 \times (P-Q). \label{kgu}
\end{align}
Let us fix a cartesian coordinate system $\{x^i\}$, then the Lie bracket reads
\[
s^i=s_1^j\p_js_2^i-s_2^j\p_j s_1^i.
\]
Note that $s_1^j\p_js_2^i(P)=\lim_{\epsilon\to 0}\frac{1}{\epsilon}[s_2^i(P+\epsilon s_{1}(P))-s_2^i(P)]$ which, using Eq. (\ref{kgu}) becomes  $s_1^j\p_js_2^i(P)=[\vec{s}_2 \times s_1(P)]^i$. Inverting the roles of $s_1$ and $s_2$ we calculate the second term, thus we obtain the interesting expression
\[
s(P)=\vec{s}_2 \times s_1(P)-\vec{s}_1 \times s_2(P).
\]
Let us check that it is a screw, indeed
\begin{align*}
s(P)-s(Q)&=\vec{s}_2 \times s_1(P)-\vec{s}_1 \times s_2(P)-\vec{s}_2 \times s_1(Q)+\vec{s}_1 \times s_2(Q)\\
&=\vec{s}_2\times [s_1(P)-s_1(Q)]-\vec{s}_1\times [s_2(P)-s_2(Q)] \\
&=\vec{s}_2\times[\vec{s}_1\times (P-Q)]-\vec{s}_1\times[\vec{s}_2\times (P-Q)]\\
&=[\vec{s}_2\cdot(P-Q))]\vec{s}_1-[\vec{s}_1\cdot(P-Q))]\vec{s}_2=(-\vec{s}_1\times \vec{s}_2)\times (P-Q),
\end{align*}
which proves also that the resultant is as claimed.
\end{proof}

The relation between the commutator and the scalar product is clarified by the following result

\begin{proposition} \label{mns}
Let $s_1,s_2,s_3$, be three screws, then
\begin{equation}
\langle s_1, [s_3,s_2]\rangle +\langle [s_3,s_1], s_2\rangle=0.
\end{equation}
Furthermore, the quantity $\langle s_1, [s_3,s_2]\rangle $ reads
\[
\langle s_1, [s_3,s_2]\rangle =s_3(P)\cdot (\vec{s}_1\times
\vec{s}_2)+s_2(P)\cdot(\vec{s}_3\times
\vec{s}_1)+s_1(P)\cdot(\vec{s}_2\times \vec{s}_3),
\]
is independent of $P$, and does not change under cyclic permutations
of its terms.
\end{proposition}

\begin{proof}
We use Eq. (\ref{jfp})
\begin{align*}
\langle s_1, [s_3,s_2]\rangle &=\vec{s}_1\cdot[\vec{s}_2\times s_3(P)-\vec{s}_3\times s_2(P)]+s_1(P)\cdot(-\vec{s}_3\times \vec{s}_2) \\
&=s_3(P)\cdot (\vec{s}_1\times \vec{s}_2)+s_2(P)\cdot(\vec{s}_3\times \vec{s}_1)+s_1(P)\cdot(\vec{s}_2\times \vec{s}_3).
\end{align*}
This expression changes sign under exchange of $s_1$ and $s_2$, thus we obtain the desired conclusion.

\end{proof}

\subsection{The dual space and the reference frame reduction to $\mathbb{R}^6$} \label{gvo}
Given a screw $s\in S$ it is possible to construct the linear map
$\langle s, \cdot \rangle: S\to \mathbb{R}$ which is an element of
the dual space $S^*$.

\begin{proposition}
The linear map $\langle s, \cdot \rangle$ sends every screw to zero
(namely, it is the null map), if and only if $s=0$.
\end{proposition}

\begin{proof}
If $s$ is such that $\vec{s}\ne 0$, then the scalar product with the
free screw $s'(P):=\vec{s}$, shows that $0=\langle s, s'
\rangle=\vec{s}^2$, a contradiction.

If $s$ is a constant screw with vector invariant $\bm{\mathcal{s}}$,
then the screw scalar product with the applied screw
$s'(P):=\bm{\mathcal{s}}\times(P-Q)$, where $Q$ is some point, gives
$0=\langle s, s' \rangle=\vec{\mathcal{{s}}}^2$, hence
$\vec{\mathcal{{s}}}=\vec{0}$ and thus $s$ is the null screw.
%
%
\end{proof}

We have shown that the linear map $s\to \langle s, \cdot \rangle$ is
injective. We wish to show that $s\to \langle s, \cdot \rangle$ is
surjective, namely any element of the dual vector space $S^{*}$, can
be regarded as the scalar product with some screw. We could deduce
this fact using the injectivity and the equal finite dimensionality
of $S$ and $S^{*}$, but we shall proceed in a more detailed way
which will allow us to introduce a useful basis for the space of
screws and its dual.

Let us choose $Q\in E$, and let  $\{\vec{e}_1,\vec{e}_2,\vec{e}_3\}$
be a positive oriented orthonormal base for $(V,\cdot,o)$, where $o$
denotes the orientation. Namely, assume that we have made a choice
of reference frame. The six screws, $f_i=((Q, \vec{e}_i),\vec{0})$,
$m_i=((Q, \vec{0}),\vec{e}_i)$, $i=1,2,3$ generate the whole space
$S$. Indeed, if $s$ is a screw and $((Q,\vec{s}),s(Q))$ is its motor
at $Q$, $\vec{s}=a_1 \vec{e}_1+a_2 \vec{e}_2+a_3 \vec{e}_3$,
$s(Q)=b_1 \vec{e}_1+b_2 \vec{e}_2+b_3 \vec{e}_3$ ,
 then $((Q,\vec{s}),s(Q))=\sum_{i=1}^3 [a_i f_i+b_i m_i]$. As a
consequence, every reference frame establishes a bijection between
the screw space $S$ and $\mathbb{R}^6$ as follows
\[
s\in S \xleftrightarrow[]{ \ reference \ frame \ } \begin{pmatrix}  \bar{a} \\
\bar{b} \end{pmatrix} \in \mathbb{R}^6
\]
where $\bar{a},\bar{b}\in \mathbb{R}^3$ (vectors in $\mathbb{R}^3$
are denoted with a bar, while the boldface notation is reserved for
vectors in $V$).

The screw scalar product between $s,s'\in S$ in this representation
takes the form
\begin{equation}
\langle s,s'\rangle= \bar{a} \cdot \bar{b}'+\bar{b}\cdot \bar{a}',
\end{equation}
thus the screw scalar product quadratic form is given by the $6
\times 6$ matrix
\begin{equation}
\begin{pmatrix} 0 & I \\I & 0\end{pmatrix} ,
\end{equation}
where $I$ is the identity $3 \times 3$ matrix.

 Let us now consider the six linear functionals  $\langle m_i,
\cdot \rangle$, $\langle f_i, \cdot \rangle$, $i=1,2,3$. From the
definition of scalar product evaluated at $Q$ it is immediate that
\begin{align*}
\langle m_i, \cdot\rangle (f_j)&=\langle m_i, f_j
\rangle=\delta_{ij}, \\
\langle f_i, \cdot\rangle (m_j)&=\langle f_i, m_j \rangle=\delta_{ij}, \\
\langle m_i, \cdot \rangle (m_j)&=\langle m_i, m_j \rangle=0, \\
\langle f_i,\cdot \rangle (f_j)&=\langle f_i, f_j \rangle=0. \\
\end{align*}
that is $\{ \langle m_i, \cdot \rangle, \langle f_i, \cdot \rangle;
i=1,2,3\}$ is the dual base to $\{f_i, m_i; i=1,2,3\}$.

 Every element $z\in S^{*}$ is uniquely determined by the
values $c_i, d_i$, $i=1,2,3$, that it takes on the six base screws
$f_i,m_i$, $i=1,2,3$. By the above formulas, the linear  combination
\[
\sum_{i=1}^3 [c_i\langle m_i, \cdot \rangle +  d_i \langle f_i,
\cdot \rangle]=\langle \sum_{i=1}^3 [c_i m_i+d_i f_i], \cdot \rangle
\]
takes the same values on the screw base and thus coincides with $z$.
We can therefore establish a bijection of the dual space $S^*$ with
$\mathbb{R}^6$ as follows
\[
z\in S^* \xleftrightarrow[]{ \ reference \  frame \ } \begin{pmatrix}  \bar{c} \\
\bar{d} \end{pmatrix} \in \mathbb{R}^6
\]
where $\bar{c},\bar{d}\in \mathbb{R}^3$ (it is convenient to
distinguish this copy of $\mathbb{R}^6$ with that isomorphic with
$S$ introduced above).

As a consequence

\begin{proposition}
The linear map $s \to \langle s, \cdot \rangle$, from $S$ to $S^*$
is bijective.
\end{proposition}

Thanks to this result any screw can be regarded either as an element
of $S$ or, acting with the screw scalar product,  as an element of
$S^*$. It must be stressed that if $s \in S$ is represented by
$\begin{pmatrix} \bar{a} \\ \bar{b} \end{pmatrix}$
then $\langle s, \cdot \rangle \in S^*$ is represented by
$\begin{pmatrix} \bar{b} \\ \bar{a} \end{pmatrix}$, that is, the map
from $S$ to $S^{*}$ which sends $s$ to $\langle s, \cdot \rangle$ is
given in this representation by the matrix $\begin{pmatrix} 0 & I
\\I & 0\end{pmatrix}$. The pairing between the elements of $S^{*}$
and those of $S$ is the usual one on $\mathbb{R}^6$. Nevertheless,
it is useful to keep in mind that we are actually in presence of two
copies of $\mathbb{R}^6$ (as we consider two isomorphisms), the
former isomorphic with $S$ and the latter isomorphic with $S^*$.

\begin{remark}
All this reduction to $\mathbb{R}^6$ depends on the reference frame.
As mentioned in the introduction most references of screw theory
introduce the screw starting from its reduction or using a
descriptive approach (the screw has an axis, a pitch, etc.). As we
argued in the introduction, it is pedagogically and logically
preferable to define the screw without making reference to  any
reference frame.
\end{remark}

For future reference we calculate, using Prop. \ref{mjc}, the
commutator between the screw base elements
\begin{align*}
[m_i,m_j]&=0,\\
[f_i,m_j]&=-[m_j,f_i]=-\sum_k \epsilon_{ijk} m_k, \\
[f_i,f_j]&=-\sum_k \epsilon_{ijk} f_k.
\end{align*}
The reader will recognize the Lie algebra commutation relations of
the group $SE(3)$ of rigid maps. We shall return to this non
accidental fact later on.

 Given a screw $s$ we consider the map $ad_s: S\to S$ which acts
as $s'\to ad_s s':=[s,s']$. Clearly, $ad_{s}s'=-ad_{s'} s$ and if
$x,y,z$ are screws, the Jacobi identity for the Lie bracket of
vector fields $[x,[y,z]]+[z,[x,y]]+[y,[z,x]]=0$,  becomes
\begin{equation} \label{bjz}
ad_{ad_x y}=ad_x ad_y-ad_y ad_x.
\end{equation}

Let an origin $O$ be given and let us use the isomorphism with
$V\oplus V$. Let $s$ be represented by  $\begin{pmatrix} \vec{s} \\
\vec{s}^{O} \end{pmatrix}$. If we introduce a full reference frame
it is possible to check with a little algebra that, according to the
above commutations, the map $ad_s$ is represented by the matrix
\begin{equation} \label{lmu}
ad_s \xleftrightarrow{origin \ O} \begin{pmatrix} -\vec{s}\times  & 0   \\
 -\vec{s}^{O} \times  &  -\vec{s} \times \end{pmatrix}
\end{equation}
where for every $\vec{v}\in V$, $\vec{v}\times:V\to V$ is an
endomorphism of $V$ induced by the vector product. Of course, if we
had kept the  reference frame $\mathbb{R}^6$ isomorphism, then, as
it is customary, with $\bar{v} \times$ we would mean  the $3 \times
3$ matrix ${\footnotesize \begin{pmatrix}
 0& -v_3 & v_2 \\ v_3 & 0 & -v_1 \\ -v_2 & v_1& 0 \end{pmatrix}}$.

\section{The kinematical screw and the composition of rigid motions}
\label{njc}

A {\em rigid motion} is a continuous map $\varphi: [0,1]\times  E
\to E$, which preserves the distances between points, i.e.\ for
every $P,Q\in E$, $t\in [0,1]$, we have
$\vert\varphi(t,P)-\varphi(t,Q)\vert=\vert P-Q\vert$, and such that
$\varphi(0,\cdot): E\to E$ is the identity map. A {\em rigid map} is
the result of a rigid motion, that is a map of type
$\varphi(1,\cdot): E\to E$. It can be shown that every rigid map is
an affine map  which preserves the scalar product and  is
orientation preserving \cite[App. 6]{maclane99}. The rigid maps form
a group usually denoted $SE(3)$.

In kinematics the velocity field of bodies performing a rigid motion
satisfies the {\em fundamental formula of the rigid body}
\begin{equation} \label{kih}
\vec{v}(P)-\vec{v}(Q)=\vec{\omega} \times (P-Q).
\end{equation}
This formula is usually deduced from Poisson formula for the time
derivative of a normalized vector: $\frac{\dd \vec{e}'}{\dd
t}=\vec{\omega} \times \vec{e}'$.

Equation (\ref{kih}) defines a screw which is  called {\em twist} in
the literature. Let us denote this screw with $k$, then
$k(P)=\vec{v}(P)$ and $\vec{k}=\vec{\omega}$, where $\vec\omega$ is
the angular velocity of the rigid body. The instantaneous axis of
rotation is by definition the screw axis of $k$.

Let us recall that if a point  moves with respect to a frame $K'$
which is in motion with respect to a frame $K$, then the velocity of
the point with respect to $K$ is obtained by summing the drag
velocity of the point, as if it were rigidly connected with frame
$K'$, with the velocity relative to $K'$. If two kinematical screws
are given and summed then the result gives a velocity field  which
represents (by interpreting one of the screw field as the velocity
field of the points at rest in $K'$ with respect to $K$) the
composition of two rigid motions. The nice fact is that the result
is independent of which screw is regarded as describing the motion
of $K'$. In other words the result has an interpretation in which
the role of the screws can be interchanged.

More generally, one may have a certain number of frames $K^{(i)}$,
$i=0,1,\ldots, n$, of which we know the screw $k_{i+1}$ which
describes the rigid motion of $K^{(i+1)}$ with respect to $K^{(i)}$.
The motion of $K^{(n)}$ with respect to $K^{(0)}$ is then described
by the screw $\sum_{i=1}^n k_{i}$. In particular, since the map
which sends a screw to its  resultant  is linear, the angular
velocity of $K^{(n)}$ with respect to $K^{(0)}$ is the sum of the
angular velocities: $\sum_{i=1}^n \vec{\omega}_{i}$. As illustrated
in the introduction, the screw approach tells us something more.
Indeed, one can establish the direction of the instantaneous axis of
rotation of $K^{(n)}$ with respect to $K^{(0)}$ by using the same
methods used to determine the central axis in a problem of applied
forces. Indeed, we shall see in a moment that there is a parallelism
between forces and angular velocities as they are both resultants of
some screw.

\section{Dynamical examples of screws}

In dynamics the most important screw is that given by the moment
field, and is called {\em wrench}. Let us recall that the momentum
$\vec{M}(Q)$ of a set of applied forces $(P_i,\vec{F}_i)$ with
respect to a point $Q$ is given by
\begin{equation} \label{cga}
\vec{M}(Q)=\sum_i (P_i-Q)\times \vec{F}_i.
\end{equation}
If we consider $P$ in place of $Q$ we get
\begin{equation} \label{cgb}
\vec{M}(P)=\sum_i (P_i-P)\times \vec{F}_i=\sum_i (P_i-Q+Q-P)\times
\vec{F}_i=\vec{M}(Q)+\vec{F}\times (P-Q),
\end{equation}
where $\vec{F}=\sum_i\vec{F}_i$ is the force resultant. This
equation shows that we are in presence of a screw $d$ such that
$d(P)=\vec{M}(P)$, $\vec{d}=\vec{F}$. The central axis of a system
of forces is nothing but the screw axis.

Another example of screw is given by the angular momentum field. The
angular momentum $\vec{L}(Q)$ of a system of point particles located
at $R_i$ with momentum $\vec{p}_i$ with respect to a point $Q$ is
given by
\begin{equation} \label{bas}
\vec{L}(Q)=\sum_i (R_i-Q)\times \vec{p}_i.
\end{equation}
If we consider $B$ in place of $Q$ we get
\[
\vec{L}(B)=\sum_i (R_i-B)\times \vec{p}_i=\sum_i (R_i-Q+Q-B)\times
\vec{p}_i=\vec{L}(Q)+\vec{P}\times (B-Q),
\]
where $\vec{P}=\sum_i\vec{p}_i$ is the total linear momentum. This
equation shows that we are in presence of a screw $l$ such that
$l(Q)=\vec{L}(Q)$, $\vec{l}=\vec{P}$.

\subsection{The cardinal equations of mechanics}
Let us consider the constitutive equation of the screw of angular
momentum
\[
\vec{L}(B)-\vec{L}(Q)=\vec{P}\times (B-Q).
\]
The vector $\vec{L}(B)$ changes in time as the distribution of
velocity and mass changes. Actually, we can consider here another
source of time change if we allow the point $B$ to change in time.
Let us first consider the case in which the angular momentum is
considered with respect to a fixed point.

By differentiating the previous equation with respect to time we get
equation (\ref{cgb}). In other words the dynamic screw $d$ is the
time derivative of the dynamic screw $l$
\begin{equation} \label{vgt}
\frac{\p l}{\p t}= d.
\end{equation}
We use here a partial derivative to remind us that the poles are
fixed.

This equation replaces the first and second cardinal equation of
mechanics. Indeed, as the map $l \to \vec{l}$ is linear it follows
\begin{equation}
\frac{\p \vec{l}}{\p t}= \vec{d},
\end{equation}
which is the first cardinal equation $\dd \vec{P}/\dd t= \vec{F}$ in
disguise. (Alternatively, write $l_t(P)=l_t(Q)+ \vec{l}_t \times
(P-Q)$ and differentiate). Here the partial derivative coincides
with the total derivative because the resultant is a free vector, it
does not depend on the point. The second cardinal equation with
respect to a point $O$
\[
\frac{\p \vec{L}(O)}{\p t}=\vec{M}(O),
\]
is obtained by evaluating Eq. (\ref{vgt}) at the  point $O$.

\subsubsection{The cardinal equation in a rigidly moving non-inertial
frame}

In Eq. (\ref{vgt}) we have differentiated with respect to time
assuming that the point with respect to which we evaluate the
angular momentum does not change in time. In other words we have
adopted a Eulerian point of view. Suppose now that on space we have
a vector field $\vec{v}(P)$ which describes the motion of a
continuum (not necessarily a rigid body). In this case we have to
distinguish the Eulerian derivative with respect to time, which we
have denoted $\p /\p t$, from the Lagrangian or total derivative
with respect to time $\dd /\dd t$. According to the latter, the
second cardinal equation of mechanics reads
\begin{equation} \label{nxg}
\frac{\dd \vec{L}(O)}{\dd t}=-\vec{v}(O)\times \vec{P}+\vec{M}(O),
\end{equation}
where $O(t)$ is the moving pole. Let us differentiate
\[
\vec{L}(B)-\vec{L}(Q)=\vec{P}\times (B-Q),
\]
with respect to time using the Lagrangian description, that is,
assuming that $B$ and $Q$ move respectively with velocities
$\vec{v}(B)$, $\vec{v}(Q)$, and considering the angular momenta with
respect to the moving points. We obtain
\begin{equation} \label{bkx}
\frac{\dd \vec{L}(B)}{\dd t}-\frac{\dd \vec{L}(Q)}{\dd
t}=\vec{F}\times (B-Q)+\vec{P}\times (\vec{v}(B)-\vec{v}(Q)).
\end{equation}
Using the second cardinal equation (\ref{nxg}) we find that this is
the constitutive equation of the momentum screw. Nevertheless, the
total derivative of the angular momentum is not a screw.

The relation between the partial and total derivative is as follows
\[
\frac{\dd \vec{L}(B)}{\dd t}=\frac{\p \vec{L}(B)}{\p
t}+\nabla_{\vec{v}(B)} \vec{L},
\]
where
\[
\nabla_{\vec{v}(B)} \vec{L}=\lim_{\epsilon\to
0}\frac{1}{\epsilon}[\vec{L}(B+\vec{v}(B)\epsilon)-\vec{L}(B)]=
\lim_{\epsilon\to 0}\frac{1}{\epsilon}[\vec{P}\times
\vec{v}(B)\epsilon ]=\vec{P}\times \vec{v}(B),
\]
thus
\begin{equation} \label{kff}
\frac{\dd \vec{L}(B)}{\dd t}=\frac{\p \vec{L}(B)}{\p
t}+\vec{P}\times \vec{v}(B).
\end{equation}
However, suppose that the velocity field is itself a screw,
\[
\vec{v}(P)-\vec{v}(R)=\vec{\omega}\times (P-R),
\]
so that the continuum moves rigidly, then from Eq. (\ref{bkx}),
using the previous results for commutators
\begin{align*}
[\frac{\dd \vec{L}(B)}{\dd t}-\,&\vec{\omega}\times
\vec{L}(B)]-[\frac{\dd \vec{L}(Q)}{\dd t}-\vec{\omega}\times
\vec{L}(Q)]\\
&=\vec{F}\times (B-Q)+[\vec{P}\times \vec{v}(B)-\vec{\omega}\times
\vec{L}(B)]-[\vec{P}\times \vec{v}(Q)-\vec{\omega}\times
\vec{L}(Q)]\\
&=\vec{F}\times (B-Q)+ [k,l](B)-[k,l](Q)=\vec{F}\times
(B-Q)+(-\vec\omega\times \vec{P})\times (B-Q)\\
&= [\vec{F}-\vec\omega\times \vec{P}]\times (B-Q).
\end{align*}
The time derivative $\frac{\dd }{\dd t})_R$ with respect to the
moving frame reads by Poisson formula, $\frac{\dd }{\dd
t})_R=\frac{\dd }{\dd t}-\,\vec{\omega}\times$, thus the previous
result can be summarized  as follows

\begin{theorem}
Let us denote with $\dd /\dd t$ the total derivative with respect to
points that move rigidly according to a kinematical screw $k$ of
vector field $\vec{v}(Q)$, and with $\frac{\dd }{\dd t})_R$ the time
derivative relative to the corresponding rigidly moving frame.
 The quantity \[\frac{\dd }{\dd t} \vec{L}(Q) \,)_R=\frac{\dd \vec{L}(Q)}{\dd
t}-\vec{\omega}\times \vec{L}(Q),\] defines a screw with resultant
$\frac{\dd }{\dd t} \vec{P} \,)_R=\vec{F}-\vec\omega\times \vec{P}$.
This screw coincides with the screw
\[
\vec{M}(Q)+[k,l](Q),
\]
thus
\[
\frac{\dd l}{\dd t})_R=\frac{\partial l}{\partial t}+[k,l].
\]

\end{theorem}

\begin{proof}
We have only to prove the last statement, which follows easily from
Eq. (\ref{kff}) and the definition of commutator.
\end{proof}

It must be remarked that in the previous result the angular momentum
$\vec{L}$ is calculated as in the original inertial frame, and not
using the point particle velocities as given in the moving frame of
twist $k$. Nevertheless, the previous result is quite interesting as
it gives a dynamical application of the commutator.

\subsection{The inertia map}
Given a rigid body the kinematical screw $k$ fixes the velocity of
every point of the rigid body and hence determines  the angular
momentum screw $l$. The map $k\to l$ is linear and is an extension
of the momentum of inertia map which includes the translational
inertia provided by  the mass.

Let us recall that given some continuum and fixed a point $Q$, the
{\em momentum of inertia map} $I: V \to V$, $\vec{\eta} \to
I_Q(\vec{\eta})$, is the linear map defined by the expression
\[
I_Q(\vec{\eta})=\sum_i m_i (R_i-Q) \times [\vec\eta\times (R_i-Q)]
\]
where we have discretized the continuum into point masses $m_i$
located, respectively, at positions $R_i$. Let $C$ be the center of
mass, namely the point defined by $\sum_i m_i (R_i-Q)=M(C-Q)$ where
$M=\sum_i m_i$. It is easy to prove the Huygens-Steiner formula
\[
I_Q(\vec{\eta})=I_C(\vec{\eta})+M (C-Q) \times [\vec\eta\times
(C-Q)].
\]

Let us consider a rigid motion described by a kinematical screw.
From Eq. (\ref{bas})
\begin{equation}
\vec{L}(Q)=\sum_i m_i (R_i-Q)\times \vec{v}(R_i),
\end{equation}
where $\vec{v}(R)$ is the kinematical screw. It is clear that this
map sends a screw into what has been proved to be another screw, and
that this map depends on the location of the masses. Let us
differentiate with respect to time the equation $\sum_i m_i
(R_i-Q)=M(C-Q)$ and use $\vec{v}(R_i)=\vec{v}(O)+\vec\omega\times
(R_i-O)$ to get
\[
\dot{C}=\frac{1}{M}\sum_i m_i \vec{v}_i=\frac{1}{M}\sum_i m_i
[\vec{v}(O)+\vec\omega\times (R_i-O)]=\vec{v}(O)+ \vec\omega\times
(C-O)=\vec{v}(C).
\]
From this equation we also obtain $\vec{P}=M\vec{v}(C)$. We have by
writing $\vec{v}(R_i)=\vec{v}(C)+\vec\omega\times (R_i-C)$
\begin{align*}
\vec{L}(Q)&=\sum_i m_i (R_i-Q)\times \vec{v}(R_i)\\
&=(C-Q)\times (M\vec{v}(C))+ \sum_i m_i (R_i-C) \times
[\vec\omega\times (R_i-C)]\\
&= (M\vec{v}(C))\times (Q-C)+ I_C(\vec\omega).
\end{align*}
This equation shows how the kinematical screw determines the
dynamical screw $l$. Plugging $Q=C$ we find
$\vec{L}(C)=I_C(\vec\omega)$, and we reobtain, as we already know,
$\vec{l}=\vec{P}=M\vec{v}(C)$.

According to Eq. (\ref{nxc}) the screw axis of the angular momentum
screw passes through the point
\[
Q=C+\frac{\vec{v}(C)\times I_C(\vec\omega)}{M v(C)^2}
\]
with direction $\vec{v}(C)$ while the instantaneous axis of rotation
passes through the point
\[
O=C+\frac{\vec{\omega}\times \vec{v}(C)}{\omega^2}
\]
and has direction $\vec{\omega}$.

%
%
%

\subsection{Screw scalar product examples: Kinetic energy, power and reciprocal screws}

Let us consider a rigid body and let us decompose it into point
particles of mass $m_i$ located at $R_i$ with mass velocity
$\vec{v}_i$ and momentum $\vec{p}_i=m_i \vec{v}_i$.

Let  $Q$ be any point
\begin{align*}
T=&\sum_i\frac{1}{2}m_i\vec{v}_i^2=\frac{1}{2}
\sum_i\vec{p}_i(\vec{v}(Q)+\vec{\omega}\times
(R_i-Q))=\frac{1}{2}\vec{P}\cdot\vec{v}(Q)+\frac{1}{2}\sum_i
\vec{p}_i\cdot(\vec{\omega}\times (R_i-Q))\\
&=\frac{1}{2}\vec{P}\cdot\vec{v}(Q)+\frac{1}{2}\sum_i
\vec{\omega}\cdot( (R_i-Q) \times
\vec{p}_i)=\frac{1}{2}[\vec{v}(Q)\cdot \vec{P}+\vec{\omega}\cdot
\vec{L}(Q)\,]=\frac{1}{2} \langle k, l \rangle.
\end{align*}
Thus the kinetic energy is one half the screw scalar product between
the kinetic screw and the angular momentum screw.

Let us now suppose that on each point  particle of mass $m_i$
located at $R_i$ acts a force $\vec{F}_i$ possibly null.
%
The power of the applied forces is the sum of the powers of
the single forces. Denoting with $L$ the work done by the forces
\begin{align*}
\frac{\dd L}{\dd t}&=\sum_i \vec{F}_i\cdot \vec{v}_i =\sum_i \vec{F}_i\cdot [\vec{v}(O)+\vec{\omega}\times (R_i-O)]\\
&=\vec{F}\cdot \vec{v}(O)+ \sum_i  \vec{F}_i\cdot
[\vec{\omega}\times (R_i-O)] =\vec{F}\cdot \vec{v}(O)+ \sum_i
\vec{\omega} \cdot [(R_i-O)\times
 \vec{F}_i]\\&
 = \vec{F}\cdot \vec{v}(O)+ \vec{\omega} \cdot \vec{M}(O)=\langle k,
 d \rangle ,
\end{align*}
that is, the total power is the screw scalar product between the
kinematical screw and the dynamical screw. Since the scalar product
is independent of  $O$, let us choose $O=C$
\begin{align*}
\frac{\dd L}{\dd t}&=\vec{F}\cdot \vec{P}/M+\vec{\omega} \cdot
\frac{\dd I_C(\vec\omega)}{\dd t}+\vec{\omega} \cdot
(\vec{v}(C)\times \vec{P})=\frac{\dd }{\dd t}
\left(\frac{\vec{P}^2}{2M}\right)+ \vec{\omega} \cdot\frac{\dd
I_C(\vec\omega)}{\dd t}\\
&=\frac{\dd }{\dd t} \left(\frac{\vec{P}^2}{2M} +\frac{1}{2}
\vec{\omega} \cdot I_C(\vec\omega) \right)+ \frac{1}{2}\vec{\omega}
\cdot\frac{\dd I_C}{\dd t} (\vec\omega)= \frac{\dd T}{\dd t}+
\frac{1}{2}\vec{\omega} \cdot\frac{\dd I_C}{\dd t} (\vec\omega),
\end{align*}
where we used K\"onig's decomposition theorem. From the kinetic
energy theorem we know that the variation of kinetic energy equals
the work done by the forces on the rigid body, thus we expect that
the last term vanishes. This is indeed so thanks to the following
lemma.

\begin{lemma}
For a rigid body $\vec{\omega}\cdot \frac{\dd I_C}{\dd t} (\vec\omega)=0$.
\end{lemma}

\begin{proof}
Let us fix a base for the vector space $V$ so that $I_C$ becomes
represented by a time dependent matrix $O^T(t)D O(t)$ where $D$ is
the diagonal  (time independent) matrix of the principal moments of
inertia and $O(t)$ is a time dependent matrix giving the rotation of
the principal directions of inertia with respect to a chosen fixed
base of $V$. Differentiating with respect to time we obtain
$\frac{\dd I_C}{\dd t} = -A I_C+I_C A$ where $A=O^T \frac{\dd O}{\dd
t}$ is an antisymmetric matrix. However, $\vec\omega$ belongs to the
kernel of this matrix from which the desired result follows.

\end{proof}

\begin{remark}
The screw is particularly useful when modeling workless constraints
between rigid bodies (think for example of a robotic arm and at its
constituent rigid parts). Indeed, suppose that  the body is made of
$N$ rigid parts and let us focus on part $i$. The constraints will
reduce the possible motions of part $i$ for a given position of the
other parts. In particular, the possible kinematical status of part
$i$ for any given relative configuration of all parts will be
described by a screw subspace $W\subset S$ of all the possible
twists of part $i$. Let $k\in W$ and let $d$ be the wrench acting on
rigid body $i$ as a result of the interaction with the neighboring
bodies. Since, by assumption, the constraints are workless we must
have $\langle d, k\rangle=0$. We conclude that the vector subspace
$Z\subset S$ made of the screws  which are screw-orthogonal  to the
allowed movements (i.e.\ screw-orthogonal to $W$), is made by all
the possible wrenches acting on body $i$ so as to make no work. Two
screws with vanishing screw scalar product (i.e.\ screw-orthogonal)
are said to be reciprocal and the subspace $Z$ is said to be
reciprocal to subspace $W$.

\end{remark}

\section{Lie algebra interpretation and Chasles' theorem} \label{mvj}

Let $g: E \to E$ be a rigid map, namely the result of a rigid motion
as it has been defined in  section \ref{njc}. We shall omit the
proof that $g$ is an affine map such that
$g(P+\vec{a})=g(P)+l(\vec{a})$ where $l:V\to V$ is a linear map
which preserves the scalar product and the orientation. The rigid
maps form a group denoted $SE(3)$ (which sometimes we shall simply
denote $G$). In the coordinates induced by a reference frame
$P'=g(P)$ has coordinates ${x^i}'$ which are related to those of $P$
by
\[
{x^i}'=\sum_j O^i_{\ j} x^j+b^i,
\]
where $O$ is a special orthogonal $3\times 3$ matrix. This
expression clarifies that $SE(3)$ is a Lie group. The three Euler
angles determining $O$ and the three translation coefficients
$\{b^i\}$, can provide a coordinate system on the 6-dimensional
group manifold. We stress that we regard $SE(3)$ as an abstract Lie
group, and we do not make any privileged choice of coordinates on it
(we do not want to make considerations that depend on the choice of
reference coordinates). The Lie algebra $\mathfrak{se}(3)$ (which
sometimes we shall simply denote  $\mathscr{G}$) of $SE(3)$ is the
family of left invariant vector fields on $SE(3)$. The Lie
commutator of two Lie algebra elements is still an element of the
Lie algebra. This structure can be identified with the tangent space
$TG_e$, $e$ being the identity element on $G$, endowed with the Lie
bracket $[,]:TG_e\times TG_e \to TG_e$.

Let $G\times E\to E$ be the above (left) group action on $E$ so that
$g_2(g_1(P))=(g_2g_1)(P)$. Each point $P\in E$ induces an {\em orbit
map} $u_P: G\to E$ given by $u_P(g)= g(P)$, thus $u_{P*}: TG_g\to
TE_{g(P)}$. We are interested on $u_{P*}$ at  $g=e$, so that
$g(P)=P$. If $v\in TG_e$ then $s(P):=u_{P*}(v)$, gives, for every
$P\in E$, a vector field on $E$ which is the image of the Lie
algebra element $v$. Such vector fields on $E$ are called {\em
fundamental vector fields}.

Let us consider the exponential $g(t)=\exp (t v)$ which is obtained
by the integration of the vector field $v$ from $e$. The orbit
$g(t)(P)$ passing through $P$ is obtained from the integration of
the vector field $s$ starting from $P$. In other words, the
1-parameter group of rigid maps $g(t)$ coincides with the
1-parameter group of diffeomorphisms (which are rigid maps)
generated by the vector field $s$. Conversely, every such
1-parameter group of rigid maps determines a Lie algebra element.
Since every screw element, once integrated, gives a non-trivial
1-parameter group of rigid maps, every screw is the fundamental
vector field of some Lie algebra element. That is, the map
$u_{P*}\vert_e: \mathscr{G} \to S$ is surjective. But  $\mathscr{G}$
and $S$ are two vector spaces of the same dimensionality, thus this
map is also injective. In summary, the screws are the representation
on $E$ of the elements of $\mathfrak{se}(3)$.

%
%
%

It is particularly convenient to study the Lie algebra of $SE(3)$
through their representative vector fields on $E$, indeed many
features, such as the existence of a screw axis for each Lie algebra
element, become very clear.

We can now use several results from the study of Lie groups and
their actions on manifolds \cite{kobayashi63,duistermaat00}. A
central result is that the bijective map $v \to s$ is linear  and
sends the Lie bracket to the commutator of vector fields on $E$. In
some cases the exponential map from the Lie algebra to the Lie group
is surjective (which is not always true as the example of
$SL(2,\mathbb{R})$ shows \cite{duistermaat00,sattinger86}). This is
the case of the group $SE(3)$ (see \cite[Prop. 2.9]{murray94}) thus,
since every element of the Lie algebra $\mathfrak{se}(3)$
corresponds to a screw $s$, and the exponential map corresponds to
the rigid map obtained from the integration of the screw vector
field by a parameter 1, the surjectivity of the exponential map
implies that every rigid map can be accomplished as the result of
the integration along a screw, or, which is the same, by a suitable
rotation along an axis combined with the translation along the same
axis. This is the celebrated Chasles' theorem \cite{selig05})
reformulated and reobtained in the Lie algebra language.

\subsection{Invariant bilinear forms and screw scalar product}

Let us recall that if $x,y\in \mathscr{G}$, then the expression
$\textrm{ad}_x y:= [x,y]$ defines a linear map $\textrm{ad}_x:
\mathscr{G}\to \mathscr{G}$ called {\em adjoint endomorphism}. The
 trace of the composition of two such endomorphisms defines a
symmetric bilinear form
\[
    K(x, y) = \textrm{trace}(\textrm{ad}_x \textrm{ad}_y),
\]
called {\em  Killing form} on $\mathscr{G}$. The Killing form is a
special {\em invariant bilinear form} $Q$ on $\mathscr{G}$, namely
it is bilinear and satisfies (use Eq. (\ref{bjz}))
\[
Q(ad_z x, y)+Q(x,ad_z y)=0.
\]
Proposition \ref{mns} shows that the scalar product of screws
provides an invariant symmetric bilinear form on the Lie algebra. We
wish to establish if there is any connection with the Killing form.

As done in section \ref{gvo} let us introduce an origin $O$ and use
the isomorphism of $S$ with $V\oplus V$. If $x, y\in S$ are
represented by $\begin{pmatrix} \vec{x} \\ \vec{x}^{O}
\end{pmatrix}$ and $\begin{pmatrix} \vec{y}\\ \vec{y}^{O}
\end{pmatrix}$,  then the screw scalar product is
\[
\langle x,y\rangle=\vec{x}\cdot \vec{y}^{O}+ \vec{x}^{O} \cdot
\vec{y}.
\]
According to the result of section \ref{gvo} we have
\begin{align*}
 K(x, y) &= \textrm{trace}\left(\begin{pmatrix} -\vec{x}\times  & 0   \\
 -\vec{x}^{O} \times  &  -\vec{x} \times \end{pmatrix}\begin{pmatrix} -\vec{y}\times  & 0   \\
 -\vec{y}^{O} \times  &  -\vec{y} \times \end{pmatrix} \right) \\
&=\textrm{trace}\begin{pmatrix} \vec{x}\times (\vec{y}\times  & 0   \\
 \vec{x}^{O} \times (\vec{y}\times + \vec{x} \times (\vec{y}^{O}\times&  \vec{x}\times (\vec{y}\times
 \end{pmatrix}\\
&=\textrm{trace}\begin{pmatrix} - (\vec{x}\cdot \vec{y}) I+ \vec{y} (\vec{x}\cdot  & 0   \\
 \vec{x}^{O} \times (\vec{y}\times + \vec{x} \times (\vec{y}^{O} \times & - (\vec{x}\cdot \vec{y}) I+  \vec{y}
  (\vec{x}\cdot
 \end{pmatrix}=-4 \vec{x}\cdot \vec{y} .
\end{align*}
We conclude that the Killing form is an invariant bilinear form
which is distinct from the screw scalar product. It coincides with
the Killing form of the Lie group of rotations alone and thus, it
does not involve the translational information inside the $O$-terms.
Therefore,  the screw scalar product provides a new interesting
invariant bilinear form, which is  sometimes referred to as the {\em
Klein form} of $\mathfrak{se}(3)$.

From Eq. (\ref{lmu}) we find that $ad_z x$ is represented by
\[
ad_z x \xleftrightarrow{origin \ O} \begin{pmatrix} \vec{x} \times \vec{z} \\
\vec{x}^{O}\times \vec{z}+\vec{x} \times \vec{z}^{O}
\end{pmatrix}
\]
from which, using the symmetry properties of the mixed product, we
can check again that the screw scalar product is invariant
\begin{align*}
\langle ad_z x,y \rangle+ \langle  x, ad_z y \rangle =& [( \vec{x}
\times \vec{z})\cdot \vec{y}^{O}+(\vec{x}^{O}\times \vec{z})\cdot
\vec{y}^{O} +(\vec{x} \times \vec{z}^{O})\cdot \vec{y}]\\ &+[(
\vec{y} \times \vec{z})\cdot \vec{x}^{O}+(\vec{y}^{O}\times
\vec{z})\cdot \vec{x} +(\vec{y} \times \vec{z}^{O})\cdot \vec{x}]=0.
\end{align*}

\section{Conclusions}

Screw theory, although venerable, has found some difficulties in
affirming itself in the curricula of the physicist and the
mechanical engineer. This has changed in the last decades, when
screw theory has finally found application in robotics, where its
ability to deal with the composition of rigid motions has proved to
be much superior with respect to treatments based on Euler
coordinates.

We have given here a short introduction to screw theory which can
provide a good starting point to a full self study of the subject.
We started from a coordinate independent definition of screw and we
went to introduce the concepts of screw axis,  screw scalar product
and screw commutator. We introduced the dual space and showed that
any frame induces an isomorphism on $\mathbb{R}^6$ which might be
used to perform calculations. We then went to consider kinematical
and dynamical examples of screws, reformulating the cardinal
equations of mechanics in this language.

Particularly important was the application of the screw scalar
product in the expressions for the kinetic energy and power, in fact
the virtual work (power) is crucial in the formulation of Lagrangian
mechanics. In this connection, we mentioned the importance of
reciprocal screws. Finally, we showed that the space of screws is
nothing but the Lie algebra $\mathfrak{se}(3)$, and that the screw
scalar product is the Klein form.

Philosophically speaking, screw theory clarifies that the most
natural basic dynamical action is not the force, but rather the
force aligned with a mechanical momenta (Remark \ref{nhs}). In
teaching we might illustrate the former action with a pushing finger
and the latter action with a kind of pushing hand. Analogously, the
basic kinematical action is not given by the act of pure rotation
(or translation) but by that of rotation aligned with translation.
Again, for illustration purposes this type of motion can be
represented with that of a (real) screw.

Clearly, in our introduction we had to omit some arguments. For
instance, we did not present neither the cylindroid nor  the
calculus of screws. Nevertheless, the arguments that we touched were
covered in full generality, emphasizing the geometrical foundations
of screw theory and its connection with the Lie group of rigid maps.
We hope that this work will promote screw theory providing an easily
accessible presentation to its key ideas.

\section*{Acknowledgments}
I thank D. Zlatanov and and M. Zoppi for suggesting some of the
cited references. This work has been partially supported by GNFM of
INDAM.

%
%
%
%
%
%
%


\begin{thebibliography}{10}

\bibitem{ball76}
Ball, R.~S.: \emph{The Theory of Screws: A study in the dynamics of
a rigid
  body}.
\newblock Dublin: Hodges, {F}oster \& {C}o. (1876)

\bibitem{ceccarelli00}
Ceccarelli, M.: Screw axis defined by {G}iulio {M}ozzi in 1763 and
early
  studies on helicoidal motion.
\newblock Mechanism and {M}achine {T}heory \textbf{35}, 761--770 (2000)

\bibitem{dimentberg68}
Dimentberg, F.~M.: \emph{The screw calculus and its applications in
mechanics}.
\newblock U.S. Department of Commerce, NTIS, AD-680 993 (1968)

\bibitem{duistermaat00}
Duistermaat, J.~J. and Kolk, J. A.~C.: \emph{Lie groups}.
\newblock Berlin: Springer (2000)

\bibitem{fasano78}
Fasano, A., {De Rienzo}, V., and Messina, A.: \emph{Corso di
Meccanica
  Razionale}.
\newblock Bari: Laterza (1978)

\bibitem{featherstone08}
Featherstone, R.: \emph{Rigid Body Dynamics Algorithms}.
\newblock New York: Springer (2008)

\bibitem{goldstein01}
Goldstein, H., Poole, C.~P., and Safko, J.~L.: \emph{Classical
mechanics (3rd
  edition)}.
\newblock Reading, Massachusetts: {Addison-Wesley} {P}ublishing {C}ompany
  (2001)

\bibitem{hunt90}
Hunt, K.~H.: \emph{Kinematic geometry of mechanisms}.
\newblock Oxford {U}niversity {P}ress (1990)

\bibitem{jazar10}
Jazar, R.~N.: \emph{Theory of applied robotics}.
\newblock New York: Springer (2010)

\bibitem{kobayashi63}
Kobayashi, S. and Nomizu, K.: \emph{Foundations of Differential
Geometry},
  vol.~I of \emph{Interscience tracts in pure and applied mathematics}.
\newblock New York: Interscience {P}ublishers (1963)

\bibitem{maclane99}
{Mac Lane}, S. and Birkhoff, G.: \emph{Algebra}.
\newblock Providence, Rhode Island: {AMS} {C}helsea {P}ublishing (1999)

\bibitem{murray94}
Murray, R.~M., Li, Z., and Sastri, S.~S.: \emph{A mathematical
introduction to
  robotic manipulation}.
\newblock Boca Raton, Florida: {CRC} {P}ress (1994)

\bibitem{sattinger86}
Sattinger, D.~H. and Weaver, O.~L.: \emph{Lie Groups and Algebras
with
  Applications to Physics, Geometry, and Mechanics}.
\newblock New York: {Springer-Verlag} (1986)

\bibitem{selig05}
Selig, J.~M.: \emph{Geometric fundamentals of robotics}.
\newblock New York: Springer (2005)

\bibitem{timoshenko65}
Timoshenko, S.~P. and Young, D.~H.: \emph{Theory of structures}.
\newblock New York: McGraw-Hill (1965)

\end{thebibliography}

\end{document}